\documentclass{sig-alternate-2013}

\newfont{\mycrnotice}{ptmr8t at 7pt}
\newfont{\myconfname}{ptmri8t at 7pt}
\let\crnotice\mycrnotice%
\let\confname\myconfname%
\usepackage{url}
\usepackage{etoolbox}
\newtoggle{PREPRINT}
\toggletrue{PREPRINT}

\iftoggle{PREPRINT}{
\permission{Disclaimer: This is the author's version of the work. It is posted
here for your personal use. Not for redistribution. The definitive version 
was published in the proceedings of the 2013 ACM SIGSAC conference on Computer \& Communications Security (CCS 2013)\\
\url{http://dl.acm.org/citation.cfm?id=2516705}
\copyrightetc{Copyright 2013 ACM \the\acmcopyr}
\crdata{978-1-4503-2477-9/13/11\ ...\$15.00.\\
http://dx.doi.org/10.1145/2508859.2516705}
}}
{
\permission{Permission to make digital or hard copies of all or part of this work for personal or classroom use is granted without fee provided that copies are not made or distributed for profit or commercial advantage and that copies bear this notice and the full citation on the first page. Copyrights for components of this work owned by others than ACM must be honored. Abstracting with credit is permitted. To copy otherwise, or republish, to post on servers or to redistribute to lists, requires prior specific permission and/or a fee. Request permissions from Permissions@acm.org.}
\conferenceinfo{CCS'13,}{November 4--8, 2013, Berlin, Gernany.}
\copyrightetc{Copyright 2013 ACM \the\acmcopyr}
\crdata{978-1-4503-2477-9/13/11\ ...\$15.00.\\
http://dx.doi.org/10.1145/2508859.2516705}
}

\usepackage{amsmath}
\usepackage{amssymb}
\usepackage{cite}
\usepackage{comment}
\usepackage[usenames]{color}
\usepackage{graphicx}
\usepackage{subfigure}

\usepackage{algorithm}
\usepackage{algorithmic}

\newcommand{\blue}{\color{blue}}

\DeclareMathOperator*{\argmax}{arg\,max}

\newtheorem{definition}{Definition}
\newdef{problem}{Problem}
\newdef{remark}{Remark}
\newtheorem{proposition}{Proposition}

\begin{document}

\title{Impact of Integrity Attacks on Real-Time Pricing in Smart Grids}

\numberofauthors{1}
\author{
\alignauthor Rui Tan$^1$ $\quad$ Varun Badrinath Krishna$^1$ $\quad$ David K. Y. Yau$^{1,2}$ $\quad$ Zbigniew Kalbarczyk$^{3}$\\
\affaddr{$^1$Advanced Digital Sciences Center, Illinois at Singapore}\\
\affaddr{$^2$Singapore University of Technology and Design, Singapore}\\
\affaddr{$^3$University of Illinois at Urbana-Champaign, Urbana, IL, USA}
\email{\{tanrui, varun.bk, david.yau\}@adsc.com.sg, kalbarcz@illinois.edu}
}

\maketitle

\begin{abstract}
  Modern information and communication technologies used by smart grids are subject to cybersecurity threats.
This paper studies the impact of integrity attacks on real-time pricing (RTP), a key feature of smart grids that uses such technologies to improve system efficiency. Recent studies have shown that RTP creates a closed loop formed by
the mutually dependent real-time price signals and price-taking demand. Such a closed loop can be exploited by an adversary whose objective is to destabilize the pricing system. Specifically, small malicious modifications to the price signals can be iteratively amplified by the closed loop, causing inefficiency and even 
severe failures such as blackouts. This paper adopts a control-theoretic approach to deriving the fundamental conditions of RTP stability under
two 
broad classes of integrity attacks, namely, the {\em scaling} and {\em delay} attacks.
We show that the RTP system is at risk of being destabilized only if the adversary can compromise the price signals advertised to smart meters by reducing their values in the scaling attack, or by
providing old prices to over half of all consumers in the delay attack.
The results provide useful guidelines for system operators to analyze the impact of various attack parameters on system stability, so that they may take adequate measures to secure RTP systems.
\end{abstract}




\section{Introduction}

  
A smart grid is an enhanced electrical grid that uses modern information and communication technologies 
to improve system reliability and efficiency. However, these computerized and networking technologies are subject to security threats that range from personal breaches~\cite{mclaughlin2010multi} to sophisticated cyber attacks launched by hostile organizations to cause widespread outages~\cite{grid-threat}.
As a sophisticated cyber-physical system, a smart grid features complex closed-loop feedback controls in various physical \cite{kundur1994power} and economic components \cite{761873}, which maintain desirable system performance in the presence of dynamics and uncertainties.
However, the impacts of cyber attacks against these closed loops on smart grids have received limited research attention. Without a systematic understanding of these impacts, system designers and operators will not be able to truly assess how these attacks may undermine the system's ability to provide mission-critical services, and hence take appropriate defensive measures against the possible threats.
This paper makes a step in
this direction
by quantifying, through both analysis and simulations, the impact of cyber attacks on real-time pricing (RTP) in smart grids, which involves closed-loop controls to stabilize the electricity market.
Dynamic pricing \cite{barbose2005real} is a widely adopted 
means to balance electricity generation and consumption.
The electricity price in the wholesale market is updated periodically (e.g., every hour) to match generation with dynamic demand.
In contrast, many current retail markets adopt static pricing schemes such as fixed and time-of-use tariffs, under which the consumers have limited incentives to adapt their electricity consumption to market conditions. This lack of incentives results in high peak demands that strain infrastructure capacities and unnecessarily increase operational costs.
By relaying the real-time wholesale prices to end users, RTP has been considered a key feature of smart grids, which can reduce over-provisioning and improve system efficiency.
Unfortunately, as analyzed in~\cite{roozbehani2011volatility}, there exists a fundamental {\em information asymmetry} between the system operators and consumers under RTP.
Specifically, a system operator needs to determine the price, which is supposed to clear the market, prior to the consumption decisions made by consumers.
As the system operator typically has limited knowledge about the consumers, its best practice is to determine the price based on historical demand.
As a result, RTP creates a closed loop formed by the
mutually dependent real-time prices and price-taking demand~\cite{roozbehani2011volatility}.
Such a closed loop can increase the system's sensitivity to dynamics and lower its robustness against situational uncertainties. As such, it can be 
exploited by an adversary whose objective is to destabilize the RTP system.
Specifically, small modifications to the signals in the closed loop made by the adversary can be iteratively amplified by the feedback,
causing inefficiency and significantly fluctuating demand, and possibly leading to severe failures such as blackouts.




In smart grids, the real-time price signals are exposed to different security threats at the source, during network transmissions, and at the consumer-level smart meters. In particular, recent
studies \cite{mclaughlin2010multi,rouf2012neighborhood} have shown that many smart meters lack basic security measures to ensure the integrity and authenticity of their input/output data. 
In light of these infrastructure vulnerabilities, imperative questions regarding RTP security include, ``Can the malicious compromise of real-time price signals destabilize the system and cause severe failures such as outages? If so, to what extent do the price signals need to be compromised?'' A main challenge in answering these questions stems from the complex coupling between the attacker actions and the closed-loop RTP system.
For instance, an attack against a few smart meters can cause monetary losses to individual victims, but it will not be able to destabilize the whole system. But if the adversary is able to compromise a sufficiently large number of consumers, the real-time price control mechanisms, which are designed to stabilize the system, may fail to mitigate the attack's impact. This impact may then pervade the whole system due to the iterative feedback. However, it is challenging to quantify these critical stability boundaries accurately, in order to characterize the impact of the attacks.

In this paper, we adopt a control-theoretic approach, which captures the closed-loop nature of the RTP, to deriving fundamental stability conditions under credible integrity attacks.
Based on the linearization of general abstract models of supply and demand, the RTP problem is formulated as a classical control problem for a linear time-invariant (LTI) system. We develop a basic pricing algorithm that sets the price adjustment proportional to the observed error between supply and demand. It ensures stability and captures the essence of stability-ensuring RTP systems. Therefore, the security analysis based on this algorithm provides a baseline understanding of the security of these systems.
We adopt a control-theoretic metric, namely, the {\em region of stability}, to characterize the security of the closed-loop RTP system with respect to important and practical adversary models. Specifically, we consider two common and broad classes of integrity attacks, which we call the {\em scaling} and {\em delay} attacks, where the prices advertised to smart meters are compromised by a scaling factor (so that the meters will use the wrong prices) and by corrupted timing information (so that the meters will use old prices), respectively.
In addition to directly tampering with traffic sent to the smart meters,
these attacks can be accomplished by indirect techniques that are less effort intensive. 
For instance, the delay attack can be realized by compromising the time synchronization of deployed smart meters. Note that current commercial smart meters~\cite{smart-meter-gps2} synchronize their clocks by either built-in GPS receivers or a network time protocol (NTP) supported by time servers~\cite{smart-meter-gps2}. Both approaches have been shown to be vulnerable to  realistic attack methods~\cite{vulnerable-ntp,nighswander2012gps}. As both attacks can be modeled as LTI transfer functions, the security analysis can be conducted under a LTI setting.

Based on our analytical framework, we derive the region of stability for both the scaling and delay attacks. In particular, we show that the RTP system will remain stable if (i) the compromised prices are amplified versions of the true prices under the scaling attack (i.e., the scaling factor exceeds one), or (ii) less than half of the consumers in the pricing system are compromised in the delay attack.
On the other hand, if the adversary can break either of these two conditions, the system may experience severely fluctuating demand arising from the system instability.
We report GridLAB-D~\cite{gridlabd} simulation results for a distribution system consisting of 1405 consumers to verify our analysis and demonstrate possible system emergencies (e.g., line and transformer overload events) caused by the integrity attacks.
Our results provide insights for securing RTP in smart grids. For instance, for the adversary to achieve the goal of compromising the price signals to at least half of the consumers, she may focus her efforts on shared support infrastructures such as the NTP time servers. This highlights the importance of securing these servers.

The rest of this paper is organized as follows. Section~\ref{sec:related} reviews related work. Section~\ref{sec:pre} presents the market model. Section~\ref{sec:rtp} defines the RTP stability problem, and develops a control-theoretic formulation of the problem. Section~\ref{sec:security} analyzes the impact of the scaling and delay attacks on the RTP system.
Section~\ref{sec:discuss} discusses extensions of the analysis framework to address a broader class of attacks that are combinations of multiple scaling and delay attacks. Section~\ref{sec:sim} presents simulation results.
Section~\ref{sec:conclude} concludes.

\section{Related Work}
\label{sec:related}


The security of smart grids is attracting increasing research attention. In particular, false data injection attacks against the state estimation of electrical grids have been extensively studied.
In \cite{liu2009injection}, Liu et al. systematically examine the conditions for bypassing a bad data detection mechanism of state estimation under various adversary capability models. Later studies \cite{yuan2011modeling,lin2012false,xie2011integrity,jia2012impacts,kosut2011malicious} show that the false data injection attacks can lead to increased system operation costs due to inordinate generation dispatch~\cite{yuan2011modeling} or energy routing~\cite{lin2012false}, as well as economic losses due to misconduct of electricity markets~\cite{xie2011integrity,jia2012impacts,kosut2011malicious}. In particular, the studies in \cite{xie2011integrity,jia2012impacts,kosut2011malicious} focus on false data injection attacks on real-time wholesale markets. They primarily emphasize attacks on critical measurements, which are often well protected by system operators.
Moreover, they ignore demand response of end users to prices.
In contrast, our work considers integrity attacks that target distributed smart meters that are much more vulnerable, and also accounts for demand response involving the end users.
All these related studies~\cite{liu2009injection,yuan2011modeling,lin2012false,xie2011integrity,jia2012impacts,kosut2011malicious} analyze attacks on systems using constrained optimization formulations such as power flow dispatch. The closed loop characterizing the RTP system in our work imposes specific challenges in the security analysis due to its iterative nature.

The security of a broader class of cyber-physical systems that feature complex closed loops has been studied recently.
In \cite{cardenas2008secure}, C{\'a}rdenas et al. identify challenges in the security analysis of these systems. In \cite{cardenas2011attacks}, the authors use simulations to study the impacts of integrity and denial-of-service attacks on a chemical reactor with multiple sensors and control loops.
In \cite{6303885}, the authors perform security threat assessment of supervisory control and data acquisition systems for water supply.
These studies focus on demonstrating the possibility of pushing the system to a certain state (e.g., unsafe pressure in a chemical reactor) by tampering with the sensor and/or control signals. They fall short of characterizing fully the fundamental critical stability conditions.

\section{Preliminaries}
\label{sec:pre}

\begin{table}
\caption{Summary of Notation*}
\label{table:symbols}
\centering
\scriptsize
\begin{tabular}{|c|l|c|}
\hline
{\bf Symbol} & {\bf Definition} & {\bf Unit} \\
\hline \hline
$T$ & pricing period & hour \\
\hline
$k$ & index of current pricing period & n/a \\
\hline
$\lambda_k$ & true price signal, $\lambda_k \in [\lambda_{\min}, \lambda_{\max}]$ & \$/MWh \\
\hline
$\lambda_k'$ & compromised price signal & \$/MWh \\
\hline
$\lambda_k^*$ & clearing price & \$/MWh \\
\hline
$b_k$ & total baseline demand & MW \\
\hline
$w(\lambda_k)$ & total price-responsive demand & MW \\
\hline
$d(k, \lambda_k)$ & total demand & MW \\
\hline
$D$ & a constant in the CEO model & MW \\
\hline
$\epsilon$ & price elasticity of demand & n/a \\
\hline
$s(\lambda_k)$ & scheduled total generation & MW \\
\hline
$\widehat{s}(\lambda_k)$ & realized total generation & MW \\
\hline
$p$ & slope of linear supply model & MW/(\$/MWh) \\
\hline
$q$ & intercept of linear supply model & MW \\
\hline
$e_k$ & generation scheduling error & MW \\
\hline
$\eta$ & price stabilization gain, $\eta \in (0,1)$ & n/a \\
\hline
$\lambda_o$ & price stabilization operating point & \$/MWh \\
\hline
$C$ & the set of all consumers & n/a \\
\hline
$C'$ & the set of consumers under attack & n/a \\
\hline
$\rho$ & $\simeq |C'| / |C|$ & n/a \\
\hline
$\gamma$ & amplification of scaling attack & n/a \\
\hline
$\tau$ & time delay of delay attack & $T$ \\
\hline
$h$ & marginal demand-supply ratio & n/a \\
\hline
\end{tabular}

* The unit of a quantity is omitted in the paper if it is specified here.
\end{table}

This section presents the market model adopted in this paper, which comprises an {\em independent system operator} (ISO) (Section~\ref{subsec:iso}), a set of consumers (Section~\ref{subsec:consumer-model}), and a set of suppliers (Section~\ref{subsec:supplier-model}).
For both consumers and suppliers, we first describe general abstract models, and then discuss concrete empirical models commonly used in literature. The analytical results in this paper (i.e., Propositions~\ref{PROP:ALGORITHM} to \ref{PROP:HARD-DELAY-ATTACK}) are based on the abstract models, while the empirical models (i.e., a constant elasticity of own-price demand model and a linear supply model) are used for the numerical examples and simulations.
The notation used in this paper is summarized in Table~\ref{table:symbols}.
We also use the following mathematical notation: $\dot{f}(x)$ denotes the first derivative of function $f(x)$; $f^{-1}(x)$ denotes the inverse of function $f(x)$; $\mathbb{R}^+$/$\mathbb{R}^-$ denotes the set of positive/negative real numbers; $\mathbb{Z}^+$ denotes the set of positive integers.

\subsection{ISO Model and RTP Schemes}
\label{subsec:iso}

The ISO is a profit-neutral agent, which aims to clear the market, i.e., match supply and demand. It determines a clearing price every $T$ hours and announces it to the suppliers and end consumers. Specifically, the price for the $k$th pricing period $[k \cdot T, (k+1) \cdot T]$, denoted by $\lambda_k$, is announced at time instant $(k-k_0) \cdot T$, where $k_0$ is a non-negative integer. Hence, this scheme corresponds to ex-ante pricing. We assume that the price must be within a range, i.e., $\lambda_k \in [\lambda_{\min}, \lambda_{\max}]$, where $\lambda_{\max} > \lambda_{\min} \in \mathbb{R}^+$. Note that in many electricity markets, suppliers sell electricity to utilities in wholesale markets, and utilities sell electricity to end consumers in retail markets. The market model adopted in this paper directly relays real-time wholesale prices to end consumers, which preserves the principles of RTP and simplifies the analysis. This model has been employed in previous studies (e.g., \cite{roozbehani2011volatility} and references therein) and is consistent with the essence of several experimental RTP programs~\cite{barbose2005real} provided by utilities, which include Board of Public Utilities in New Jersey, Baltimore Gas and Electric Company in Maryland, and Duquesne Light in Pennsylvania. In these programs, the hourly wholesale prices published by PJM Interconnection LLC\footnote{PJM is a Regional Transmission Organization (RTO). This paper does not distinguish between ISO and RTO.}
are used directly as retail prices, where $T=1$ and $k_0=0$. A few other experimental RTP programs give customers advance notice of hourly prices. For instance, $k_0=1$ for the RTP-HA-2 program of Georgia Power~\cite{rtp-ha-2}. To simplify the discussion, we focus on RTP schemes without advance notice, i.e., $k_0=0$. However, our analysis can be easily extended to encompass advance notice. In reality, locational prices can be applied to address location-dependent transmission costs. In many areas, as generation cost dominates transmission cost, variations of locational prices are often small. For instance, the relative standard deviation of the locational prices for 219 locations published by PJM is often around 5\% only~\cite{pjm-price}. As this paper focuses on the impact of integrity attacks on RTP systems, we ignore the small variations in the locational prices. Thus, we assume that all the suppliers and consumers are subject to the same real-time price $\lambda_k$.

\subsection{Consumers}
\label{subsec:consumer-model}



\noindent {\bf Abstract demand model:}
Let $C$ denote the set of consumers in the system. For consumer $j \in C$, let $b_{k,j}$ denote the {\em baseline demand} in the $k$th pricing period, which is exogenous, bounded, dependent on time, but independent of $\lambda_k$. For instance, for a household, the baseline demand can characterize the minimum necessary power usage, such as cooking and a minimum level of illumination.
Let $v_j(x - b_{k,j})$ denote the additional value (unit: \$/hour) derived from consuming a total of $x$ units of power in the $k$th pricing period, where $x \ge b_{k,j}$. We assume that $v_j(\cdot)$ is a strictly increasing and strictly concave function. Let $d_j(k, \lambda_k)$ denote the demand of consumer $j$ in the $k$th pricing period given price $\lambda_k$. We denote $w_j(\lambda_k) = \dot{v}_j^{-1}(\lambda_k)$, where $\dot{v}_j^{-1}(\cdot)$ represents the inverse function of $\dot{v}_j(\cdot)$.
The demand is given by
\begin{equation*}
d_j(k, \lambda_k) = \argmax\nolimits_{x} \left( v_j(x \! - \! b_{k,j}) \! - \! \lambda_k \cdot x \right) = b_{k,j} \! + \! w_j(\lambda_k).
\end{equation*}
In the above equation, $v_j(x - b_{k,j}) - \lambda_k \cdot x$ is the additional utility beyond the baseline by consuming totally $x$ units of power. The $w_j(\lambda_k)$ is referred to as {\em price-responsive demand}. It is easy to verify that $w_j(\lambda_k)$ is a decreasing function of $\lambda_k$.
By denoting $b_k = \sum_{j \in C} b_{k,j}$ and $w(\lambda_k) = \sum_{j \in C} w_j(\lambda_k)$, the total demand, denoted by $d(k, \lambda_k)$, is given by
\begin{equation}
d(k, \lambda_k) = \sum\nolimits_{j \in C} d_j(k, \lambda_k) = b_k + w(\lambda_k).
\label{eq:demand3}
\end{equation}
As there are a large number of consumers, we assume that $b_{k,j}$ and $v_j(\cdot)$ are unknown to the ISO. However, the ISO knows the historical total demand $\{ d(h, \lambda_h) | h \in [0, k-1] \}$.
The above derivations, which are based on the basic concept of {\em utility} in economics, explain the consumer's demand response to price. Human-induced demand response has been observed in previous studies~\cite{sweeny2002california}. With the increasing adoption of smart appliances and home automation systems, this demand response will become more automated.

\vspace{0.5em}
\noindent {\bf Empirical demand model:}
The constant elasticity of own-price (CEO) model~\cite{fleten2005constructing}
is a simple model that can be used
to characterize the total price-responsive demand, which is defined by
  $w(\lambda_k) = D \cdot \lambda_k^\epsilon$,
where $D$ and $\epsilon$ are positive and negative constants, respectively. The $\epsilon$ is referred to as the {\em price elasticity of demand}, which is typically within $(-1,0)$~\cite{filippini2011short,lijesen2007real}.

\subsection{Suppliers}
\label{subsec:supplier-model}


\noindent {\bf Abstract supply model:}
Each supplier aims to maximize its profit. Let $S$ denote the set of suppliers in the system. For any supplier $i \in S$, let the function $c_i(x)$ represent the cost (unit: \$/hour \cite[p.~534]{grainger1994power}) of producing and transmitting $x$ units of power. We assume that $c_i(x)$ is a strictly convex and non-negative function over the support $x \ge 0$. Moreover, we assume that $c_i(x)$ is an asymptotically increasing function, i.e., $\exists x_0 \ge 0$, $\dot{c}_i(x) > 0$ if $x > x_0$. Let $s_i(\lambda_k)$ denote the quantity of power that supplier $i$ schedules to generate in the $k$th pricing period given price $\lambda_k$, which is given by $s_i(\lambda_k) = \argmax_{x} \left( \lambda_k \cdot x - c_i(x) \right) = \dot{c}_i^{-1}(\lambda_k)$. Note that $\lambda_k \cdot x - c_i(x)$ is the profit from generating $x$ units of power.
It is easy to verify that $s_i(\lambda_k)$ is an increasing function of $\lambda_k$. We assume that the generation capacity of the supplier $i$ is at least $s_i(\lambda_{\max})$. In this paper, we consider centralized bulk generation rather than distributed generation.
Therefore, an ISO can estimate $\{c_i(\cdot) | \forall i \in S \}$ as there are typically a limited number of suppliers.
Let $s(\lambda_k)$ denote the {\em scheduled} total supply in the $k$th pricing period, i.e., $s(\lambda_k) = \sum_{i \in S} s_i(\lambda_k)$. We note that, in current electricity wholesale markets, the supply and price are often determined through a bidding process~\cite{fleten2005constructing}, which is generally governed by the costs of generation and transmission. In a competitive bidding-based wholesale market, the resultant supply and price will well reflect the supply model derived from the cost model.
We assume that the {\em realized} total generation in the $k$th pricing period, denoted by $\widehat{s}(\lambda_k)$, is always equal to the total demand $d(k, \lambda_k)$. This is consistent with the current technologies in power grids. For instance, when the demand exceeds the scheduled generation, the system operators will observe a dropping voltage and frequency, and generation can be increased to meet demand and maintain the voltage and frequency at their nominal values.

\begin{figure}
  \centering
  \includegraphics{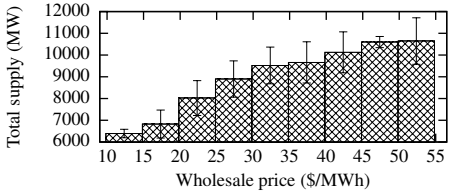}
  \caption{Total supply vs. wholesale price \cite{AEMO}.}
  \label{fig:supply-price}
\end{figure}

\vspace{0.5em}
\noindent {\bf Empirical supply model:}
Quadratic cost functions have been widely adopted in the analyses of power generation systems~\cite{grainger1994power}, i.e., $c_i(x) = \alpha_i x^2 + \beta_i x + \delta_i$, where $\alpha_i > 0$. To make $c_i(x)$ non-negative over $x \ge 0$, we have a few additional conditions: $\delta_i \ge 0$ if $\beta_i \ge 0$, or $\delta_i \ge \frac{\beta_i^2}{4\alpha_i}$ if $\beta_i < 0$. Therefore, if $\lambda_k \ge \beta_i$, $s_i(\lambda_k) = \frac{1}{2\alpha_i} (\lambda_k - \beta_i)$; otherwise, $s_i(\lambda_k) = 0$.
To simplify the evaluation based on this empirical supply model, we assume that $\beta_i < 0$ and $\delta_i \ge \frac{\beta_i^2}{4\alpha_i}$, such that the total supply can be simplified as
$s(\lambda_k) =  p \cdot \lambda_k + q$,
where $p = \sum_{i \in S}\frac{1}{2\alpha_i} > 0$ and $q = \sum_{i \in S} -\frac{\beta_i}{2\alpha_i} > 0$.
We now empirically verify this linear supply model using the half-hourly total supply data of New South Wales (NSW), Australia, provided by the Australian Energy Market Operator (AEMO)~\cite{AEMO}. Fig.~\ref{fig:supply-price} shows the histogram of total supply versus the wholesale price in January, 2012. We can see that the relationship between the average supply and price is nearly linear. A linear fitting of the total supply shown in Fig.~\ref{fig:supply-price} yields $p=152$ and $q=4503$. Such a linear relationship can also be seen in the investigation of the electricity market of California~\cite[p.~112]{sweeny2002california}, where the demand does not exceed the generation capacity.

\section{The RTP Problem and Solutions}
\label{sec:rtp}


This section formally states the RTP problem, examines an existing solution, and proposes a new basic control-theoretic solution with provable {\em bounded-input bounded-output stability} (referred to as {\em stability} for short in this paper). Based on our solution, the security analysis in Section~\ref{sec:security} lays the foundation for understanding the impact of attacks on feedback-based RTP systems.


\subsection{The RTP Problem and Solution Stability}
\label{subsec:RTP-problem}

\begin{figure}
  \centering
  \includegraphics{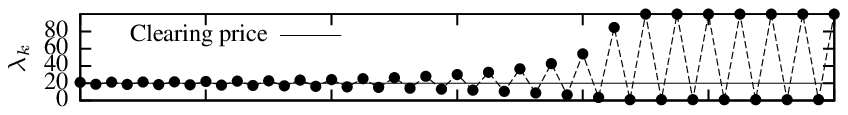}
  \includegraphics{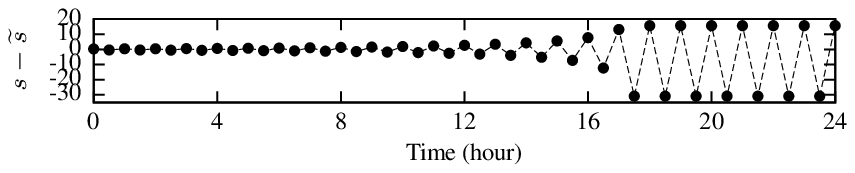}
  \vspace{-0.2in}
  \caption{An example of unstable solution \cite{roozbehani2011volatility} under the linear supply and CEO demand models. Top figure: Evolution of price. Bottom figure: Generation scheduling error in GW. ($p$=152, $q$=4503, $\epsilon=-0.6$, $b_k$=2000, $\lambda_k^*$=20, $\lambda_0$=21, $\lambda_{\min}$=1, $\lambda_{\max}$=100)}
  \label{fig:unstable-example}
\end{figure}

At time instant $k\cdot T$, the ISO aims to find the {\em clearing price} for the period $[k \cdot T, (k+1) \cdot T]$, denoted by $\lambda_k^*$, such that the scheduled supply matches demand, i.e., $s(\lambda_k^*) = d(k, \lambda_k^*)$.
However, as $d(k, \lambda_k)$ is unknown, in practice, the ISO sets the price to match the scheduled supply and {\em predicted} demand (denoted by $\widetilde{d}(k, \lambda_k)$). Formally, we define

\vspace{.5em}
{\em RTP Problem:} Find $\lambda_k$ such that $s(\lambda_k) = \widetilde{d}(k, \lambda_k)$.
\vspace{.5em}

Straightforward solutions to this problem may lead to significantly fluctuating prices. For instance, a direct feedback approach~\cite{roozbehani2011volatility}, which uses $d(k-1, \lambda_{k-1})$ as the predicted demand $\widetilde{d}(k, \lambda_k)$, can yield oscillating prices as shown in Fig.~\ref{fig:unstable-example}. The root cause of the oscillation is the unstable closed-loop system formed by the direct feedback.
When the system is unstable, the price set by the ISO will oscillate or diverge, even if the initial price is very close to the true clearing price. The oscillations may lead to severe consequences.
When the diverging prices reach low values, the increased demand may cause overload of the transmission and distribution networks. Moreover, as shown in Fig.~\ref{fig:unstable-example}, the unstable system may experience significant {\em generation scheduling errors} (i.e., $s(\lambda_k) - \widehat{s}(\lambda_k)$). Although reserve generating capacity can help compensate for the errors, their use may increase the cost of operating the system.




To study the impact of integrity attacks on the RTP systems, we should start with RTP schemes that are stable in the absence of attacks. In Section~\ref{subsec:direct-feedback}, we examine the stability of an existing RTP scheme \cite{roozbehani2011volatility} in the absence of attacks. Its poor stability 
properties motivate us to design a basic control-theoretic RTP scheme with provable stability in the absence of attacks. This is the subject of Section~\ref{sec:classical-control}.

\subsection{Direct Feedback Approach}
\label{subsec:direct-feedback}

A direct feedback approach to the RTP problem has been studied in \cite{roozbehani2011volatility}. The conditions for global stability\footnote{The system is {\em globally stable} if the price converges to the clearing price given any positive initial price.} of the approach, i.e., the properties of $s(\cdot)$ and $d(\cdot, \cdot)$ that ensure global stability, have also been analyzed in \cite{roozbehani2011volatility}.
The approach is briefly reviewed as follows. It predicts $d(k, \lambda_k)$ by the most recent demands based on an autoregression model, and determines the price $\lambda_k$ accordingly. For instance, the simplest autoregression model uses $d(k-1, \lambda_{k-1})$ as $\widetilde{d}(k, \lambda_k)$ and the closed-loop system is expressed as $s(\lambda_k) = d(k-1, \lambda_{k-1})$. It is also referred to as the {\em persistence model}. Thus, the price is determined as $\lambda_k = s^{-1}(d(k-1, \lambda_{k-1}))$. If direct feedback based on the persistence model is not globally stable, it is difficult to stabilize those systems globally with an autoregression-based direct feedback approach \cite{roozbehani2011volatility}. Hence, global stability under the persistence model is particularly important. By applying Corollary~3 in \cite{roozbehani2011volatility}, our analysis \cite{thistechreport} shows that
under the linear supply and CEO demand models, where $\epsilon \in (-1,0)$ and $b_k$ is a non-negative constant $b$,
the system $\lambda_k = s^{-1}(d(k-1, \lambda_{k-1}))$ is not globally stable.

\begin{figure}
  \subfigure[Map of convergence probability ($b=0$).]{
    \includegraphics{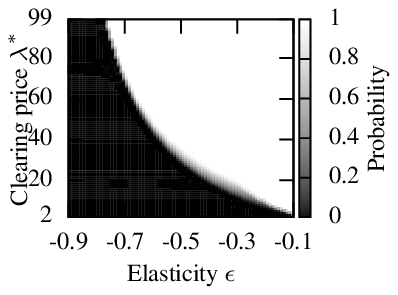}
    \label{fig:stability-prob}
  }
  \hspace{.2em}
  \subfigure[Stability boundary under various settings of $b$.]{
    \includegraphics{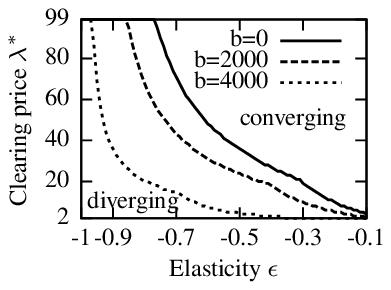}
    \label{fig:stability-boundary-df}
  }
  \vspace{-0.15in}
  \caption{Stability of direct feedback approach under the linear supply and CEO demand models.}
  \label{fig:empirical-stability-study}
  \vspace{-0.1in}
\end{figure}

As the direct feedback approach is not globally stable, its convergence highly depends on the system state. If $b_k$ is time-varying, it can push the system to a state that eventually leads to divergence. A few realistic constraints may affect the system stability. For instance, even if the system is not globally stable, the system may converge when the initial price is within the allowed range $[\lambda_{\min}, \lambda_{\max}]$. Moreover, if a tentative price is out of the range $[\lambda_{\min}, \lambda_{\max}]$, it will be rounded to $\lambda_{\min}$ or $\lambda_{\max}$. Hence, we conduct numerical experiments that account for these realistic constraints for better understanding.
The settings of the supply model are $p=152$ and $q=4503$, which are obtained in Fig.~\ref{fig:supply-price}.
Given a clearing price $\lambda^* \in [\lambda_{\min}, \lambda_{\max}]$, the coefficient $D$ is set by solving $s(\lambda^*) = d(\lambda^*)$, i.e., $D = \frac{p \lambda^* + q - b}{\lambda^{*\epsilon}}$, where $b \in [0, p \lambda_{\min} + q)$ to ensure $D > 0$ for any valid $\lambda^*$.
Fig.~\ref{fig:stability-prob} shows a map of the probability that the system is converging when $b=0$. To calculate the probability, the initial price sweeps the range $[\lambda_{\min}, \lambda_{\max}]$ and the probability is calculated as the fraction of the initial prices that lead to system convergence. Fig.~\ref{fig:stability-prob} shows that the probability is mostly either 0 or 1 and the transition region with the probability within $(0,1)$ is sharp. Fig.~\ref{fig:stability-boundary-df} plots the boundaries between the converging and diverging regions under various settings of $b$, where its valid range is $[0, 4503)$.
For instance, when $b=4000$, the system can be diverging if $\epsilon=-0.8$ and $\lambda^*<20$. For the data shown in Fig.~\ref{fig:supply-price}, about 20\% of the prices are lower than $20$. Therefore, the direct feedback approach can be unstable with significant probabilities.


\subsection{Control-Theoretic Price Stabilization}
\label{sec:classical-control}

The results in Section~\ref{subsec:direct-feedback} show the necessity of control laws for stabilizing the RTP systems.
This section develops a basic control-theoretic price stabilization algorithm with provable stability.
The main objective of this paper is to identify the fundamental impacts of integrity attacks against the vulnerable real-time price signals on the stability of the RTP systems with well designed control laws.
More sophisticated price stabilization algorithms could be developed.
However, our security analysis in Section~\ref{sec:security}, which is based on our basic control-theoretic algorithm, provides fundamental baselines for understanding the security properties of RTP systems running such sophisticated algorithms.

The objective of price stabilization is to minimize the generation scheduling error and adapt to the time-varying baseline load.
We reformulate the RTP problem as a classical discrete-time feedback control problem. Under this formulation, the ISO observes the generation scheduling error in the previous pricing period, and then uses it to guide the setting of the price in the next pricing period. Specifically, let $e_k$ denote the generation scheduling error, i.e., $e_k = s(\lambda_k) - d(k, \lambda_k)$. The objective is to maintain the {\em controlled variable} $e_k$ close to its {\em reference}, which is zero. The {\em manipulated variable} is $\lambda_k$, and $s(\lambda_k) - d(k, \lambda_k)$ is the {\em controlled system}. The block diagram of the feedback control loop is shown in Fig.~\ref{fig:classical-control}. We let $G_c(z)$, $G_p(z)$, and $H(z)$ denote the transfer functions of the price stabilization algorithm, the controlled system, and the observation system, which are expressed in the $z$-transform domain. The $z$-transform~\cite{ogata1995discrete} provides a compact representation for discrete-time functions, where $z$ represents a time shift operation. As $b_k$ is bounded and independent of $\lambda_k$, it can be modelled as a {\em disturbance} to the system \cite{ogata1995discrete}.

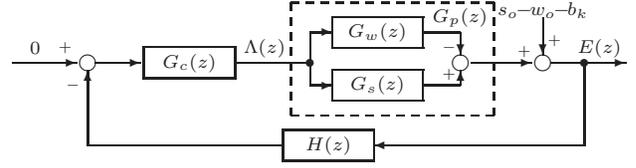
\begin{figure}
  \scriptsize
  \setlength{\unitlength}{1cm}
  \centering
  \begin{picture}(10,2.3)(0,1.5)
    \put(0.3,3.2){\makebox(0,0){0}}
    \put(0,3){\vector(1,0){0.85}}
    \put(1,3){\circle{0.2}}    
    \put(0.7,3.2){\makebox(0,0){$ \scriptstyle + $}}
    \put(0.8,2.7){\makebox(0,0){$ \scriptstyle - $}}
    \put(1.1,3){\vector(1,0){0.6}}
    \put(1.75,2.8){\framebox(1.2,0.4){$G_c(z)$}}
    \put(2.95,3){\line(1,0){1}}
    \put(3.35,3.2){\makebox(0,0){$\Lambda(z)$}}
    \put(3.95,3){\circle*{0.1}}
    \put(3.95,3){\line(0,1){0.4}}
    \put(3.95,3){\line(0,-1){0.3}}
    \put(3.95,3.4){\vector(1,0){0.3}}
    \put(3.95,2.7){\vector(1,0){0.3}}
    \put(4.25,3.2){\framebox(1.2,0.4){$G_w(z)$}}
    \put(4.25,2.5){\framebox(1.2,0.4){$G_s(z)$}}
    \put(5.45,3.4){\line(1,0){0.5}}
    \put(5.45,2.7){\line(1,0){0.5}}
    \put(5.95,3){\circle{0.2}}
    \put(5.8,2.85){\makebox(0,0){$ \scriptstyle +$}}
    \put(5.8,3.2){\makebox(0,0){$ \scriptstyle -$}}
    \put(5.95,3.4){\vector(0,-1){0.3}}
    \put(5.95,2.7){\vector(0,1){0.2}}
    \put(6.1,3){\vector(1,0){0.8}}
    \put(7.05,3.6){\vector(0,-1){0.5}}
    \put(7.05,3.7){\makebox(0,0){$s_o \!\! - \!\! w_o \!\! - \!\! b_k$}}
    \put(7.05,3){\circle{0.2}}
    \put(6.8,3.15){\makebox(0,0){$ \scriptstyle +$}}
    \put(7.2,3.3){\makebox(0,0){$ \scriptstyle +$}}
    \put(7.15,3){\vector(1,0){1}}
    \put(7.8,3.2){\makebox(0,0){$E(z)$}}
    \multiput(3.7,3.8)(0.2,0){14}{\line(1,0){0.1}}
    \multiput(3.7,2.3)(0.2,0){14}{\line(1,0){0.1}}
    \multiput(3.7,2.3)(0,0.2){8}{\line(0,1){0.1}}
    \multiput(6.4,2.3)(0,0.2){8}{\line(0,1){0.1}}
    \put(5.95,3.6){\makebox(0,0){$G_p(z)$}}
    \put(7.6,3){\circle*{0.1}}
    \put(7.6,3){\line(0,-1){1.1}}
    \put(7.6,1.9){\vector(-1,0){2.8}}
    \put(3.6,1.7){\framebox(1.2,0.4){$H(z)$}}
    \put(3.6,1.9){\line(-1,0){2.6}}
    \put(1,1.9){\vector(0,1){1}}
  \end{picture}
  \vspace{-0.2in}
  \caption{The control-theoretic price stabilization. $\Lambda(z)$ and $E(z)$ are the $z$-transforms of $\lambda_k$ and $e_k$.}
  \label{fig:classical-control}
\end{figure}

We now derive the expressions of $G_p(z)$ and $H(z)$. To preserve generality, our design is based on the abstract supply and demand models $s(\lambda)$ and $w(\lambda)$, which can be non-linear as in the CEO model. In controller design, a common approach to dealing with non-linear systems is to adopt local linearization~\cite{ogata1995discrete}.
Specifically, $s(\lambda) \simeq s(\lambda_o) + \dot{s}(\lambda_o) \cdot (\lambda - \lambda_o)$ and $w(\lambda) \simeq w(\lambda_o) + \dot{w}(\lambda_o) \cdot (\lambda - \lambda_o)$, where $\lambda_o$ is a {\em fixed operating point}. By denoting $s_p(\lambda) = \dot{s}(\lambda_o) \lambda$, $s_o = s(\lambda_o) - \dot{s}(\lambda_o) \lambda_o$, $w_p(\lambda) = \dot{w}(\lambda_o) \lambda$, and $w_o = w(\lambda_o) - \dot{w}(\lambda_o) \lambda_o$, we have $s(\lambda) \simeq s_p(\lambda) + s_o$ and $w(\lambda) \simeq w_p(\lambda) + w_o$. As $s_o$ and $w_o$ are independent of $\lambda$, as shown in Fig.~\ref{fig:classical-control}, we can collect them with the price-independent $b_k$. The transfer functions of the proportional models $s_p(\lambda)$ and $w_p(\lambda)$ are $G_s(z) = \dot{s}(\lambda_o)$ and $G_w(z) = \dot{w}(\lambda_o)$, respectively. Therefore, $G_p(z) = G_s(z) - G_w(z) = \dot{s}(\lambda_o) - \dot{w}(\lambda_o)$.
As the price stabilization algorithm uses the observed generation scheduling error in the previous pricing period to adjust the price for the current pricing period, $H(z)=z^{-1}$, which represents the delay of one pricing period.
Based on the above modeling, we have the following proposition, which can be proved by examining whether the poles of the system are located within a unit circle centered at the origin of the $z$-plane. The details of the proof are omitted due to space constraints and can be found in \cite{thistechreport}.

\begin{proposition}
For the linearized system $G_p(z)=\dot{s}(\lambda_o)-\dot{w}(\lambda_o)$ with $\lambda_o$ fixed and the observation system $H(z)=z^{-1}$, the following price stabilization algorithm ensures stability: $\lambda_k = \lambda_{k-1} - \frac{2\eta}{\dot{s}(\lambda_o) - \dot{w}(\lambda_o)} \cdot e_{k-1}$, where $\eta \in (0,1)$.
\label{PROP:ALGORITHM}
\end{proposition}
The transfer function of the above algorithm is $G_c(z)=\frac{2\eta}{(\dot{s}(\lambda_o) - \dot{w}(\lambda_o))(1 - z^{-1})}$. From control theory, when $b_k$ is a constant, the system converges the fastest when $\eta=0.5$, as the system's pole is at the origin \cite{ogata1995discrete}. The convergence speed is particularly important for adapting to fast time-varying baseline load so that the convergence is achieved before a significant change of baseline load. However, our analysis in Section~\ref{sec:security} shows that we generally need to set a smaller $\eta$ to reduce the impact of attacks. In other words, we have to sacrifice convergence speed for resilience to attacks.

As discussed in Section~\ref{subsec:consumer-model}, $w(\cdot)$ is unknown to the ISO. In practice, the ISO can estimate $\dot{w}(\lambda_o)$ based on the history of price-demand pairs. Our analysis in Section~\ref{subsec:inaccurate-w} shows that, if the relative error in estimating $\dot{w}(\lambda_o)$ is less than $100\times(1-\eta) \%$, the algorithm given by Proposition~\ref{PROP:ALGORITHM} remains stable. For instance, if $\eta=0.5$, the relative error bound is 50\%, which is a tractable requirement for most estimation algorithms. Moreover, for a smaller $\eta$ that is set to increase resilience to attacks, the error bound will be larger. As the focus of this paper is to analyze the fundamental impact of integrity attacks on system stability under the control law in Proposition~\ref{PROP:ALGORITHM}, we do not elaborate on the estimation algorithm, and the analysis in Section~\ref{sec:security} assumes that the ISO can accurately estimate $\dot{w}(\lambda_o)$. Section~\ref{subsec:inaccurate-w} also discusses the impact of inaccurate $\dot{w}(\lambda_o)$ on the security analysis.

The price stabilization algorithm in Proposition~\ref{PROP:ALGORITHM} assumes a fixed operating point $\lambda_o$. However, intuitively, if the operating point $\lambda_o$ adapts to the current price, the linear approximations to $s(\lambda)$ and $w(\lambda)$ can be more accurate. Specifically,
by setting $\lambda_o = \lambda_{k-1}$, we have the following algorithm:
\begin{equation}
\lambda_k = \lambda_{k-1} - \frac{2 \eta}{\dot{s}(\lambda_{k-1}) - \dot{w}(\lambda_{k-1})} \cdot e_{k-1}.
\label{eq:classical-control}
\end{equation}
Although there is a lack of rigorous theory to support the technique of adapting $\lambda_o$ to the current price, our numerical experiments show that the algorithm in Eq.~(\ref{eq:classical-control}) is always stable under all the settings shown in Fig.~\ref{fig:empirical-stability-study}. The numerical examples and simulations conducted in the rest of this paper employ the algorithm in Eq.~(\ref{eq:classical-control}). Fig.~\ref{fig:running-example}(a) shows the evolution of price with fixed baseline load. When $\eta=0.5$, $\lambda_k$ converges to $\lambda^*$ after two pricing periods. When $\eta=0.2$, the system has a longer settling time. When $\eta=0.8$, the price oscillates but converges. The oscillation is caused by a negative pole \cite{ogata1995discrete}.
Fig.~\ref{fig:running-example} will also be used as a running example in Section~\ref{sec:security} to illustrate the impact of attacks.

\begin{figure}
\includegraphics{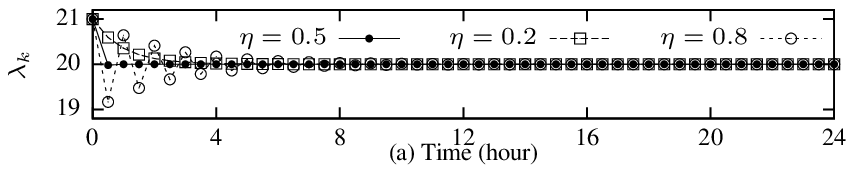}
\includegraphics{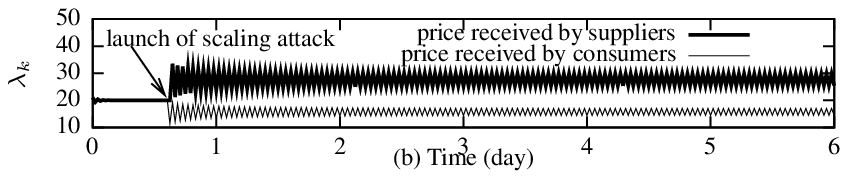}
\includegraphics{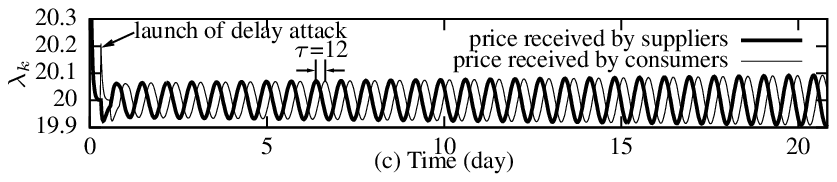}
\caption{A running example under the linear supply and CEO demand models ($T=0.5$, $\epsilon=-0.8$, $b_k=2000$, $\lambda^*=20$, $\lambda_0=21$). (a) Price stabilization; (b) Scaling attack ($\eta=0.8$, $\rho=1$, $\gamma=0.57$); (c) Delay attack ($\eta=0.2$, $\rho=1$, $\tau=12$).}
\label{fig:running-example}
\end{figure}

\section{Integrity Attacks to RTP}
\label{sec:security}

This section studies the impact of two integrity attacks on RTP systems under the RTP scheme given by Proposition~\ref{PROP:ALGORITHM}.

\subsection{Attack Models and Impact Metrics}

We consider integrity attacks on the price signals received by a subset of consumers. If the price signal received by a consumer is subject to attack, the price signal applied for the current pricing period (denoted by $\lambda_k'$) is different from the true price $\lambda_k$.
The integrity attacks on the price signals can be launched in different ways. For instance, once the adversary has compromised the intermediate nodes in the communication network of the smart grid (e.g., routers) and obtained the decryption/encryption keys held by the ISO and/or smart meters, the adversary can intercept and forge price data packets. Moreover, recent reverse engineering and penetration tests~\cite{mclaughlin2010multi,rouf2012neighborhood} have shown that many smart meters lack basic security measures to ensure integrity and authenticity of the input/output data. These security vulnerabilities can be exploited to maliciously change the price signals. We would like to point out that the integrity attacks do pose strong requirements for the adversary. They require that the adversary is able to modify the price information, either at the source, during transmissions, or at the smart meters. 
However, these attacks in a cyber environment are certainly feasible and credible, and it would be wrongfully complacent to ignore their possibility.  
In this paper, as the price signals sent to the centralized suppliers are often well protected, we assume that they are not subject to attacks. However, our analysis framework can be easily extended to account for possible attacks on the suppliers.

\subsubsection{Attack Models}
\label{subsec:attack-models}


As the number of consumers in a smart grid is often large, the number of compromised consumers is an important metric for the adversary's capability and resource availability. Let $C'$ denote the set of consumers whose price signals are compromised, where $C' \subseteq C$, and $w'(\lambda_k)$ denote the total price-responsive demand in the presence of an attack. Thus, $w'(\lambda_k) = \sum_{j \in C'} w_j(\lambda_k') + \sum_{j \in C \setminus C'} w_j(\lambda_k)$. We define
\begin{equation}
\rho = \frac{\sum_{j \in C'} w_j(\lambda_k')}{\sum_{j \in C} w_j(\lambda_k')} = \frac{\sum_{j \in C'} w_j(\lambda_k')}{w(\lambda_k')},
\label{eq:rho}
\end{equation}
which characterizes the fraction of consumers receiving the compromised price signals. If the consumers are {\em homogeneous} (i.e., $w_j(\cdot)$ is same for all $j$), $\rho$ is a constant, i.e., $\rho=|C'|/|C|$. If they are heterogeneous, $\rho$ is a function of $\lambda_k'$. The extensive numerical evaluation in \cite{thistechreport} shows that, if the heterogeneous consumers follow the CEO model, $\rho \simeq |C'|/|C|$ with a variation of less than 0.003 and hence can be practically treated as a constant. Moreover, we make the following approximation:
\begin{equation}
\sum\nolimits_{j \in C \setminus C'} \! w_j(\lambda_k) \simeq (1 \! - \! \rho) \sum\nolimits_{j \in C} \! w_j(\lambda_k) = (1 \! - \! \rho) w(\lambda_k).
\label{eq:demand-approximate}
\end{equation}
The numerical evaluation in \cite{thistechreport} shows that relative approximation error of Eq.~(\ref{eq:demand-approximate}) is less than 1\%. Therefore, in the presence of integrity attacks, we have
\begin{equation}
w'(\lambda_k) \simeq \rho w(\lambda_k') + (1 - \rho) w(\lambda_k).
\label{eq:demand-under-attack}
\end{equation}

If the price signals can be arbitrarily modified, the capability requirements of an adversary would be high.
In this paper, we consider ``constrained'' integrity attacks, where the malicious modifications follow certain rules and can be realized with lower capability and resource requirements.
Note that the adversary must be able to cause more severe damage to the system if she is assumed to be able to modify the price signals arbitrarily.
An attack can be characterized by the parameters for the rule, which is denoted by $\mathcal{A}$. We consider two kinds of integrity attacks:

\vspace{.5em}
\noindent {\bf Scaling attack $\mathcal{A}=(\rho, \gamma)$:} The compromised price is a scaled version of the true price, i.e., $\lambda_k' = \gamma \lambda_k$, $\gamma \in \mathbb{R}^+$.

\vspace{.5em}
\noindent {\bf Delay attack $\mathcal{A}=(\rho, \tau)$:} The compromised price is an old price, i.e., $\lambda_k' = \lambda_{k-\tau}$, $\tau \in \mathbb{Z}^+$.
\vspace{.5em}


These two attacks can be launched in various ways. The price values or time stamps in data packets sent to the smart meters can be maliciously modified during transmissions in  vulnerable communication networks. Moreover, they can be launched in indirect ways. For instance, the delay attack can be launched by modifying the smart meters' internal clocks. Smart meters typically assign a memory buffer to store received prices. If a smart meter's clock has a lag, it will store newly received prices in the buffer and apply an old price for the present. Furthermore, attacks on the clocks can be realized by compromising vulnerable time synchronization protocols or the time servers that provide timing information to the smart meters. A few smart meter products~\cite{smart-meter-gps2} synchronize their clocks via a built-in GPS receiver, which is vulnerable and subject to remote attacks that are effective across large geographic areas~\cite{nighswander2012gps}.


In this paper, we assume that at most one kind of attack is in effect.
Moreover, we assume that the attack parameters are the same for all the compromised consumers. For instance, if a delay attack with $\tau=2$ is launched, all the compromised consumers experience the same delay of two pricing periods. These simplifications allow us to better understand the impact of each attack on the RTP system, which is the basis for understanding more complex scenarios such as heterogeneous attack parameters and combinations of attack types. In Section~\ref{subsec:superimpose}, we will briefly discuss how to extend our analysis to address these more complex cases.

\subsubsection{Attack Impact Metrics}

This section defines two metrics for the impact of the integrity attacks on system stability. We first define the {\em marginal demand-supply ratio}, which is a quantity that can significantly affect the system stability under attacks.
\begin{definition}
Marginal demand-supply ratio is $h = \left|\frac{\dot{w}(\lambda_o)}{\dot{s}(\lambda_o)}\right|$.
\label{def:incentive-ratio}
\end{definition}
From Definition~\ref{def:incentive-ratio}, $h$ depends on the operating point $\lambda_o$. As discussed in Section~\ref{sec:classical-control}, the gain coefficient $\eta$ of the price stabilization algorithm affects the system stability in a major way. Therefore, we define the following metric:
\begin{definition}
  Given attack $\mathcal{A}$, the {\em region of operating point stability} under attack, denoted by $\mathrm{ROS}_{\lambda_o}(\mathcal{A})$, is 
\begin{equation*}
\mathrm{ROS}_{\lambda_o}(\mathcal{A}) = \{ (h, \eta) | \text{The system is stable under attack $\mathcal{A}$} \}.
\end{equation*}
\end{definition}
The above metric depends on $\lambda_o$. We define a second metric that is independent of $\lambda_o$:
\begin{definition}
  Given attack $\mathcal{A}$, the {\em region of stability} under attack, denoted by $\mathrm{ROS}(\mathcal{A})$, is 
\begin{equation*}
\mathrm{ROS}(\mathcal{A}) = \{ \eta | \text{The system is stable under attack $\mathcal{A}$}, \forall h > 0\}.
\end{equation*}
\end{definition}


The above two metrics are important for understanding the impact of integrity attacks on the stability of the RTP system under the price stabilization algorithm in Proposition~\ref{PROP:ALGORITHM}. In particular, the $\mathrm{ROS}(\mathcal{A})$ specifies the range of $\eta$ that ensures system stability under attack $\mathcal{A}$. Hence, the ROS allows us to compare the impacts of different integrity attacks. For two attacks $\mathcal{A}_1$ and $\mathcal{A}_2$, if $\mathrm{ROS}(\mathcal{A}_1) \subset \mathrm{ROS}(\mathcal{A}_2)$, the ISO has more flexibility in setting $\eta$ under $\mathcal{A}_2$ than $\mathcal{A}_1$, to achieve faster convergence. Thus, the system is more resilient to $\mathcal{A}_2$ than $\mathcal{A}_1$. From the adversary's perspective, $\mathcal{A}_1$ is more effective than $\mathcal{A}_2$. Note that, when the RTP system with $\eta \in \mathrm{ROS}(\mathcal{A})$ is stable under attack $\mathcal{A}$, the compromised consumers may still experience monetary losses and the system may run at low efficiency. However, this paper focuses on the impact of attacks on the system stability, which is a fundamental system requirement.
In Sections~\ref{subsec:resilience-scaling} and \ref{subsec:resilience-delay}, we will derive the $\mathrm{ROS}_{\lambda_o}$ and $\mathrm{ROS}$ for the scaling and delay attacks.

\subsection{Impact of Scaling Attack}
\label{subsec:resilience-scaling}


The local linearization of Eq.~(\ref{eq:demand-under-attack}) with $\lambda_k'=\gamma \lambda_k$ is
\begin{align*}
w'(\lambda_k) \simeq & \rho \cdot ( w(\gamma \lambda_o) + \dot{w}(\gamma \lambda_o) \cdot (\gamma \lambda_k - \gamma \lambda_o) )\\
& + (1-\rho) \cdot ( w(\lambda_o) + \dot{w}(\lambda_o) \cdot (\lambda_k - \lambda_o) ).
\end{align*}
By collecting the price-independent terms with $b_k$, the transfer function of the price-dependent component is $G_w(z) = \rho \gamma \dot{w}(\gamma \lambda_o) + (1 - \rho) \dot{w}(\lambda_o)$. To make the analysis tractable, for the scaling attack, we only focus on a class of price-responsive demand models that satisfy $\dot{w}(\gamma \lambda) = \dot{w}(\lambda) \cdot \mu(\gamma | \Theta)$, where $\Theta$ is the set of model parameters of $w(\cdot)$ and the function $\mu(\gamma | \Theta)$ is independent of $\lambda$ and always positive. Such a $\dot{w}(\cdot)$ is said to be {\em decomposable}. For instance, under the CEO model, $\Theta=\{D, \epsilon\}$, and $\mu(\gamma | \Theta) = \gamma^{\epsilon-1}$. For simplicity of exposition, we denote $\mu(\gamma | \Theta)$ as $\mu$ in the rest of this paper. Therefore, $G_w(z) = \rho \gamma \mu \dot{w}(\lambda_o) + (1 - \rho) \dot{w}(\lambda_o)$, and $G_p(z) = G_s(z) - G_w(z) = \dot{s}(\lambda_o) - \rho \gamma \mu \dot{w}(\lambda_o) - (1 - \rho) \dot{w}(\lambda_o)$. The closed-loop transfer function \cite{ogata1995discrete} under the attack is
\begin{equation*}
T_c(z) = \frac{G_c(z) G_p(z)}{1 \! + \! G_c(z) G_p(z) H(z)} \! = \! \frac{2 \eta ( 1 \! + \! \rho \gamma \mu h \! + \! h \! - \! \rho h)z}{P(z)},
\end{equation*}
where the {\em system characteristic function} $P(z)= (h + 1)(z - 1) + 2\eta (1 + \rho \gamma \mu h + h - \rho h)$. Note that $G_p(z)$, $H(z)$, and $G_c(z)$ have been obtained in Section~\ref{sec:classical-control}.

\subsubsection{Region of Operating Point Stability}

\begin{proposition}
For the linearized system based on a fixed operating point $\lambda_o$ and a decomposable $\dot{w}(\cdot)$, $\mathrm{ROS}_{\lambda_o}(\rho, \gamma) = \{ (h, \eta) | 0 < \eta < \min \{ 1,  \bar{\eta} \}, \forall h > 0 \}$, where
\begin{equation}
\bar{\eta} = \frac{h+1}{h + 1 + \rho h (\gamma \mu - 1)}.
\label{eq:bar-eta}
\end{equation}
\label{PROP:SCALING-ROS}
\end{proposition}
\begin{proof}
  If all the poles of $T_c(z)$ (i.e., roots of $P(z)$) are within the unit circle centered at the origin of $z$-plane, the system is stable \cite{ogata1995discrete}. If $\eta < \bar{\eta}$, the pole is within the circle. As $\eta \in (0,1)$, $\eta$ takes the minimum of 1 and $\bar{\eta}$.
\end{proof}

\begin{remark}
Under the CEO demand model, by replacing $\mu = \gamma^{\epsilon - 1}$ in Eq.~(\ref{eq:bar-eta}), we have $\bar{\eta} = \frac{h+1}{h + 1 + \rho h (\gamma^\epsilon - 1)}$. Fig.~\ref{fig:scaling-point-ros} plots the stability boundaries when $\epsilon=-0.8$, where the $\mathrm{ROS}_{\lambda_o}$ are the regions below the boundaries. We can see that the $\mathrm{ROS}_{\lambda_o}$ shrinks with increased $\rho$ and decreased $\gamma$. This can be easily proved by the monotonicity of $\bar{\eta}$. Moreover, it is consistent with the intuitions that (i) the system becomes more unstable when more consumers are compromised, and (ii) the increased demand due to a decreased $\gamma$ poses more challenges to the system.
\end{remark}

We now use the numerical example in Fig.~\ref{fig:running-example}(b) to verify our analysis. Fig.~\ref{fig:running-example}(b) shows the price signals received by the suppliers and consumers, respectively, when $\gamma=0.57$. We can see that the price does not converge. The average value of $h$ is $0.850$, which falls in the unstable region ($h > 0.786$) according to the analytical $\mathrm{ROS}_{\lambda_o}$. Note that when $\gamma=0.59$, the price converges and the average value of $h$ is $0.862$, which falls in the stable region ($h < 0.908$) according to the analytical $\mathrm{ROS}_{\lambda_o}$. Therefore, Proposition~\ref{PROP:SCALING-ROS} successfully characterizes the critical stability boundary. Note that, as the settings for Fig.~\ref{fig:running-example}(b) are close to the stability boundary, the price oscillates in a small range. For smaller $\gamma$, the price can severely oscillate, as shown in Section~\ref{sec:sim}.

\begin{figure}
  \subfigure[$\rho=1$]{
    \includegraphics{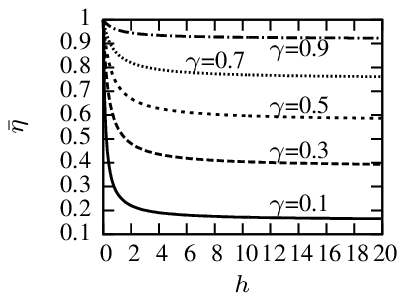}
  }
  \subfigure[$\gamma=0.5$]{
    \includegraphics{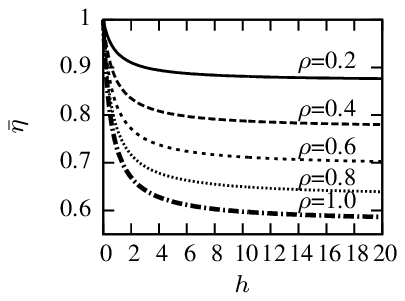}
  }
  \caption{Stability boundaries under scaling attack, abstract supply model, and CEO demand model ($\epsilon=-0.8$, $\mathrm{ROS}_{\lambda_o}$ are the regions below the boundaries).}
  \label{fig:scaling-point-ros}
\end{figure}

\subsubsection{Region of Stability}

\begin{proposition}
For the linearized system based on a decomposable $\dot{w}(\cdot)$, when $\gamma \mu \in (0, 1]$, $\mathrm{ROS}(\rho, \gamma) = \{ \eta | 0 < \eta < 1\}$; when $\gamma \mu > 1$, $\mathrm{ROS}(\rho, \gamma) = \left\{ \eta \left| 0 < \eta < \inf_{h > 0} \bar{\eta} \right. \right\}$, where $\inf_{h > 0} \bar{\eta} = \frac{1}{1 + \rho ( \gamma \mu - 1) }$.
\label{PROP:SCALING-ROAS}
\end{proposition}
\begin{proof}
When $\gamma \mu \in (0, 1]$, $\bar{\eta} \ge 1$. From Proposition~\ref{PROP:SCALING-ROS}, if $0 < \eta < 1$, the system is stable regardless of $h$. When $\gamma \mu > 1$, $\bar{\eta}$ is a bounded decreasing function of $h$. Its infimum $\inf_{h > 0} \bar{\eta} = \lim_{h \to +\infty} \bar{\eta} = \frac{1}{1 + \rho ( \gamma \mu - 1) }$. Therefore, if $0 < \eta < \inf_{h>0}\bar{\eta}$, the system is stable regardless of $h$.
\end{proof}

\begin{remark}
Under the CEO demand model, replacing $\mu = \gamma^{\epsilon-1}$ in Proposition~\ref{PROP:SCALING-ROAS} yields the following result. When $\gamma \ge 1$, $\mathrm{ROS}(\rho, \gamma) = \{ \eta | 0 < \eta < 1\}$; when $\gamma \in (0,1)$, $\mathrm{ROS}(\rho, \gamma) = \left\{ \eta \left| 0 < \eta < \inf_{h > 0} \bar{\eta} \right. \right\}$, where $\inf_{h > 0} \bar{\eta} = \frac{1}{1 + \rho ( \gamma^\epsilon - 1) }$. Therefore, under the CEO model, if the adversary amplifies the price, the system remains stable. This result is consistent with the intuition that decreased demand due to the amplified price poses no challenges to the system. Fig.~\ref{fig:scaling-limit} plots $\inf_{h > 0} \bar{\eta}$ when $\epsilon=-0.8$. We can see that $\mathrm{ROS}$ shrinks with increased $\rho$ and decreased $\gamma$. This can be easily proved by the monotonicity of $\inf_{h > 0} \bar{\eta}$.
\end{remark}

\subsection{Impact of Delay Attack}
\label{subsec:resilience-delay}

The local linearization of Eq.~(\ref{eq:demand-under-attack}) with $\lambda_k' = \lambda_{k - \tau}$ is
\begin{align*}
w'(\lambda_k) \simeq &\rho \cdot (w(\lambda_o) + \dot{w}(\lambda_o) \cdot (\lambda_{k-\tau} - \lambda_o)) \\
&+ (1-\rho) \cdot (w(\lambda_o) + \dot{w}(\lambda_o) \cdot (\lambda_k - \lambda_o)).
\end{align*}
By collecting the price-independent terms with $b_k$, the transfer function of the price-dependent component is $G_w(z)=z^{-\tau} \rho \dot{w}(\lambda_o) - (1-\rho) \dot{w}(\lambda_o)$, where $z^{-\tau}$ represent a delay of $\tau$ pricing periods. Therefore, $G_p(z) = G_s(z) - G_w(z) = \dot{s}(\lambda_o) - z^{-\tau} \rho \dot{w}(\lambda_o) - (1-\rho) \dot{w}(\lambda_o)$.
The closed-loop transfer function under the attack is
\begin{equation*}
T_c(z) \! = \! \frac{G_c(z) G_p(z)}{1 \! + \! G_c(z) G_p(z) H(z)} \! = \! \frac{2\eta (1 \! + \! (1 \! - \! \rho)h) z^{\tau \! + \! 1} \! + \! 2 \rho \eta h z}{P(z)},
\end{equation*}
where the system characteristic function is
\begin{equation*}
P(z) = (h+1) z^{\tau + 1} + \left( 2 \eta + 2\eta(1-\rho)h - h - 1 \right) z^\tau + 2\eta \rho h.
\end{equation*}

\subsubsection{Region of Operating Point Stability}

\begin{figure}
  \centering
  \subfigure[$\rho=1$]{
    \includegraphics[width=.222\textwidth]{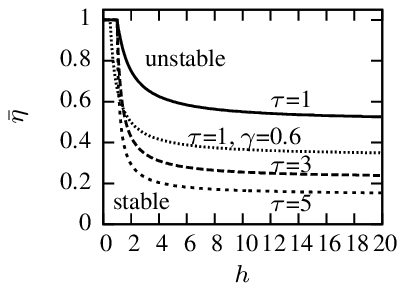}
    \label{fig:stability-boundary1}
  }
  \subfigure[$\tau=4$]{
    \includegraphics[width=.222\textwidth]{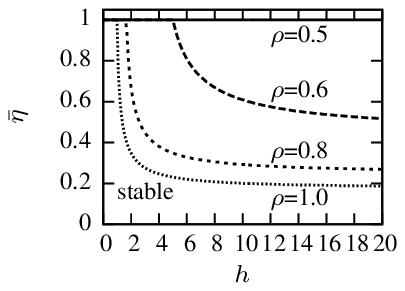}
    \label{fig:stability-boundary2}
  }
  \caption{Stability boundaries under delay attack and abstract supply/demand models. (In the left figure, the curve with $\tau=1$ and $\gamma=0.6$ is for combined attack discussed in Section~\ref{subsec:superimpose}.)}
  \label{fig:stability-boundary}
\end{figure}

As $P(z)$ is a $(\tau+1)$-order polynomial, it is extremely difficult to derive the closed-form formulas for the poles of $T_c(z)$.
Various methods have been developed to test the stability without explicitly solving for the poles \cite{ogata1995discrete}. Among them, the Jury test~\cite[p.~185]{ogata1995discrete} is preferred because the coefficients of $P(z)$ are real numbers. The Jury test constructs a table based on the coefficients of $P(z)$ and derives the stability conditions from the table. Given $\rho$, we can derive the closed-form $\mathrm{ROS}_{\lambda_o}$ for different $\tau$ from the Jury test.
However, the expressions become more complicated for larger $\tau$.
We numerically compute the $\mathrm{ROS}_{\lambda_o}$ based on the Jury test for various settings of $\tau$ and $\rho$. Fig.~\ref{fig:stability-boundary} plots the stability boundaries, where the $\mathrm{ROS}_{\lambda_o}$ are the regions below the boundaries. From Fig.~\ref{fig:stability-boundary}, the $\mathrm{ROS}_{\lambda_o}$ shrinks with $\tau$ and $\rho$, which is consistent with intuition.
We have the following proposition. The proof is based on the Jury test,
which is omitted due to space constraints and can be found in \cite{thistechreport}.
\begin{proposition}
For the linearized system with a fixed operating point $\lambda_o$, $\mathrm{ROS}_{\lambda_o}(\rho, \tau + 1) \subseteq \mathrm{ROS}_{\lambda_o}(\rho, \tau)$.
\label{PROP:ROS}
\end{proposition}



We now use the numerical example in Fig.~\ref{fig:running-example} to verify our analysis. Fig.~\ref{fig:running-example}(c) shows the price signals received by the suppliers and consumers, respectively, when $\eta=0.2$, $\rho=1$, $\tau=12$. We can see that the price diverges.
The average value of $h$ is $1.455390$, which falls in the unstable region ($h >1.447$) according to the Jury test. Note that when $\tau=11$, the price does converge and the average value of $h$ is $1.455335$, which falls in the stable region ($h < 1.522$) according to the Jury test. Therefore, the Jury test successfully characterizes the critical stability boundary. As the settings for Fig.~\ref{fig:running-example}(c) are close to the stability boundary, the price diverges slowly. For larger $\tau$, the price can diverge quickly.

\subsubsection{Region of Stability}

We observe from Fig.~\ref{fig:stability-boundary2} that, when $\rho=0.5$, the system is stable for $\eta \in (0,1)$. We have the following proposition.
\begin{proposition}
For the linearized system, if $\rho \in (0, 0.5]$, $\forall \tau \in \mathbb{Z}^+$, $\mathrm{ROS}(\rho, \tau) = \{\eta | 0 < \eta < 1 \}$.
\label{PROP:HARD-DELAY-ATTACK}
\end{proposition}
The proof can be found in the appendix, where we prove that if $\rho \in (0,0.5]$, all roots of $P(z)$ are within the unit circle centered at the origin in the $z$-plane and hence the system is stable~\cite{ogata1995discrete}. From Proposition~\ref{PROP:HARD-DELAY-ATTACK}, to launch a successful delay attack that destabilizes the system, the adversary has to compromise no less than a half of the consumers. The intuition behind this result is that the compromised price-responsive load must predominate to affect the operation of the system.
This result poses strong requirements for the adversary. However, she could accomplish the goal by targeting shared infrastructures such as the time servers that provide timing information to all the smart meters. On the other hand, the need for the adversary to compromise a large fraction of the meters in order to be effective is indicative of the resilience of the price stabilization algorithm given by Proposition~\ref{PROP:ALGORITHM} to delay attacks.


\renewcommand{\algorithmicrequire}{\textbf{Input:}}
\renewcommand{\algorithmicensure}{\textbf{Output:}}
\renewcommand{\algorithmiccomment}[1]{// #1}
\begin{algorithm}[t]
\caption{Compute $\mathrm{ROS}(\rho, \tau)$ when $\rho \in (0.5,1]$}
\label{alg:symbolic}
\begin{algorithmic}[1]
\scriptsize
\REQUIRE $\rho$ and $\tau$
\ENSURE $\lim_{h \to +\infty} \bar{\eta}(h | \rho, \tau)$
\IF{$\tau=1$}
\RETURN $1/(2\cdot \rho)$
\ENDIF
\STATE $\mathbf{X} = 1$, $\mathbf{Y} = 2\eta\rho$, $\mathbf{Z} = 2\eta\rho(1-2\eta(1-\rho))$
\label{ln:1}
\FOR{$i=1$ to $\tau$}
\STATE $\mathbf{U} = \mathbf{X} \cdot \mathbf{X} - \mathbf{Y} \cdot \mathbf{Y}$, $\;$ $\mathbf{V} = \mathbf{Z}$, $\;$ $\mathbf{W} = (\mathbf{X} \cdot \mathbf{Z} \cdot \mathbf{Z}) / \mathbf{Y}$
\STATE $\mathbf{X} = \mathbf{U}$, $\;$ $\mathbf{Y} = \mathbf{V}$, $\;$ $\mathbf{Z} = \mathbf{W}$
\ENDFOR
\STATE $\mathbf{Q}(\eta | \rho) = \mathbf{X} \cdot \mathbf{X} - \mathbf{Y} \cdot \mathbf{Y} - \mathbf{Z}$
\label{ln:7}
\RETURN minimum root of $\mathbf{Q}(\eta | \rho) = 0$ over $\eta \in (0,1)$
\label{ln:8}
\end{algorithmic}
\vspace{-0.5em}
\line(1,0){240}\\
\scriptsize Note: Line~\ref{ln:1} to Line~\ref{ln:7} are symbolic calculation, where the bold capitals are symbolic expressions of $\eta$ and $\rho$.
\end{algorithm}

We now discuss the ROS when $\rho \in (0.5,1]$. From Fig.~\ref{fig:stability-boundary}, the stability boundary curves are non-increasing and converge to limits when $h \to +\infty$. Let $\bar{\eta}(h | \rho, \tau)$ denote the stability boundary curve for particular $\rho$ and $\tau$. Therefore,
\begin{equation*}
\mathrm{ROS}(\rho, \tau) = \{ \eta | 0 < \eta < \lim\nolimits_{h \to +\infty} \bar{\eta}(h | \rho, \tau) \}.
\end{equation*}
When $\tau = 1$, the limit is simply $\frac{1}{2\rho}$. However, for larger $\tau$, it is extremely difficult to derive the closed-form formula for the limit, primarily because of the iterative nature of the Jury test. Instead, we use an algorithm to define $\mathrm{ROS}(\rho, \tau)$, which is shown in Algorithm~\ref{alg:symbolic}. This algorithm is developed based on  key observations from the Jury test procedure.
Fig.~\ref{fig:jury-limit} plots $\lim_{h \to +\infty} \bar{\eta}(h | \rho, \tau)$, which is computed by Algorithm~\ref{alg:symbolic}, versus $\tau$ under various settings of $\rho$. We also use the Jury test to compute $\bar{\eta}(h | \rho, \tau)$ with a large setting for $h$ (specifically, $h=10^{10}$). The results are the same as in Fig.~\ref{fig:jury-limit}. From the figure, we can see that the ROS shrinks with $\rho$ and $\tau$, which is consistent with intuition.

\begin{figure}
  \centering
  \subfigure[Scaling attack (abstract supply model and CEO demand model with $\epsilon$=-0.8)]{
    \includegraphics{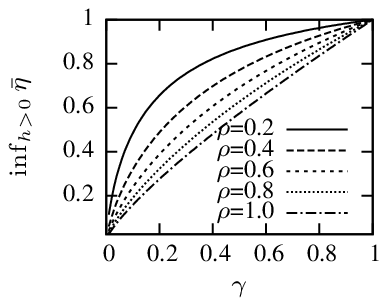}
    \label{fig:scaling-limit}
  }
  \hspace{.05in}
  \subfigure[Delay attack (abstract supply and demand models)]{
    \includegraphics{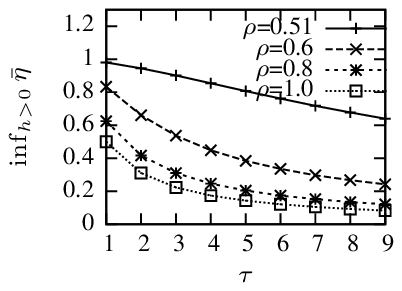}
    \label{fig:jury-limit}
  }
  \vspace{-0.15in}
  \caption{Upper bound of ROS.}
\end{figure}

\section{Discussions}
\label{sec:discuss}

In this section, we discuss the impact of inaccuracy in estimating $\dot{w}(\lambda_o)$ on the analysis in the previous sections. We also discuss how to extend our analysis to address more complicated attack models.

\subsection{Impact of Inaccuracy in Estimating $\dot{w}(\lambda_o)$}
\label{subsec:inaccurate-w}

As the price-dependent demand model $w(\cdot)$ is unknown to the ISO, we derive an upper bound for the error in estimating $\dot{w}(\lambda_o)$, to ensure the stability of the algorithm in Proposition~\ref{PROP:ALGORITHM}. Let $\widetilde{\dot{w}(\lambda_o)}$ denote the estimated $\dot{w}(\lambda_o)$, and $E_w \!\! = \!\! \frac{\dot{w}(\lambda_o) \! - \! \widetilde{\dot{w}(\lambda_o)}}{\dot{w}(\lambda_o)}$ denote the relative estimation error. The stability condition $0 < \eta < 1$ can be rewritten as $0 < \frac{2\eta}{\dot{s}(\lambda_o)-\dot{w}(\lambda_o)} < \frac{2}{\dot{s}(\lambda_o) - \dot{w}(\lambda_o)}$. As long as $0 < \frac{2\eta}{\dot{s}(\lambda_o)-\widetilde{\dot{w}(\lambda_o)}} < \frac{2}{\dot{s}(\lambda_o) - \dot{w}(\lambda_o)}$, the system is stable. This condition can be derived as $E_w \! < \! (1 \! - \! \eta)\left(1 \! - \! \frac{\dot{s}(\lambda_o)}{\dot{w}(\lambda_o)}\right)$. As $1 \! - \! \frac{\dot{s}(\lambda_o)}{\dot{w}(\lambda_o)} \! > \! 1$, $E_w \! < \! (1 \! - \! \eta)$ is a sufficient condition for stability.


We now discuss the impact of inaccurate $\widetilde{\dot{w}(\lambda_o)}$ on the security analysis results in Section~\ref{sec:security}. From the definition of $E_w$, we have $\widetilde{\dot{w}(\lambda_o)} = ( 1 - E_w) \dot{w}(\lambda_o)$. Note that $E_w < 1$ since $\widetilde{\dot{w}(\lambda_o)} > 0$. By replacing $h$ with $|1 - E_w| \cdot h$ in Proposition~\ref{PROP:SCALING-ROS}, we have a new result in the presence of the estimation error $E_w$. From the proofs of Propositions~\ref{PROP:SCALING-ROAS}, \ref{PROP:ROS}, and \ref{PROP:HARD-DELAY-ATTACK}, they are independent of $h$. Therefore, these propositions still hold in the presence of estimation errors.

\subsection{Superimposed and Heterogeneous Attacks}
\label{subsec:superimpose}

In this section, we discuss how to extend our analysis framework to address a class of integrity attacks that are the superimposition of scaling and delay attacks.
We also discuss how to adapt our analysis to scenarios in which the attack models/parameters are different for different compromised consumers.


From discrete-time control theory \cite{ogata1995discrete}, our analysis framework can be applied to derive the $\mathrm{ROS}_{\lambda_o}$ and $\mathrm{ROS}$ under any integrity attack that can be modeled as a linear time-invariant (LTI) system with the transfer function $\frac{\Lambda'(z)}{\Lambda(z)} = \frac{\sum_{i=0}^{n}a_iz^{-i}}{\sum_{j=0}^{m}b_jz^{-j}}$, where the $\Lambda(z)$ and $\Lambda'(z)$ are the $z$-transforms of $\lambda_k$ and $\lambda_k'$. In the time domain, $\lambda_k'$ is given by the linear combination of $\lambda_{k - i}$ and $\lambda'_{k-j}$, where $0 \le i \le n$ and $1 \le j \le m$.
The scaling and delay attacks studied in Section~\ref{sec:security} are special cases of this general attack model. For instance, under the delay attack, $b_0=1$, $b_j=0$ for $j \ge 1$, $a_\tau = 1$, $a_i = 0$ for $i \neq \tau$. This general attack model can also be regarded as the superimposition of scaling and delay attacks.
We now illustrate the enhanced impact of attack superimposition using a simple example: $\lambda_k' = \gamma \cdot \lambda_{k-\tau}$. Under this attack superimposition,
the closed-loop system characteristic function of Eq.~(\ref{eq:demand-under-attack}) is
\begin{equation*}
P(z) = (h+1) z^{\tau + 1}  + (2\eta + 2 \eta (1-\rho)h - h -1) z^{\tau} + 2\eta \rho h \gamma \mu,
\end{equation*}
where $\mu$ is defined in Section~\ref{subsec:resilience-scaling}. We can still apply the Jury test to derive the $\mathrm{ROS}_{\lambda_o}(\rho, \gamma, \tau)$ and $\mathrm{ROS}(\rho, \gamma, \tau)$. Fig.~\ref{fig:stability-boundary1} shows the stability boundary for this attack superimposition with $\rho=1$, $\gamma=0.6$, and $\tau=1$. The $\mathrm{ROS}_{\lambda_o}$ of this attack superimposition is smaller than the delay attack with $\rho=1$ and $\tau=1$, which means stronger attack impact.

If two subsets of consumers are subject to two different attacks that happen simultaneously in the grid,
Eq.~(\ref{eq:demand-under-attack}) can be rewritten as $w'(\lambda_k)
= \rho_1w(\lambda_k') + \rho_2w(\lambda_k'') + (1 - \rho_1 - \rho_2)
w(\lambda_k)$, where $\rho_1$ and $\rho_2$ are the fractions of
consumers subject to the two attacks, and $\lambda_k'$ and
$\lambda_k''$ are the corresponding compromised prices. Our analysis framework still applies once the models of $\lambda_k'$ and $\lambda_k''$ are specified.
The attack with different parameters (e.g., consumers are subject to different delays) can be treated as simultaneous attacks.
\section{Trace-Driven Simulations}
\label{sec:sim}

We use GridLAB-D \cite{gridlabd}, an electric power distribution system simulator, to evaluate the impact of integrity attacks.
GridLAB-D captures many physical characteristics such as power line capacities and impedances.
Hence, we can validate our analysis under the realism provided by GridLAB-D. Moreover, it can record emergency events that occur when the current ratings of lines and power ratings of transformers are exceeded. Such events could cause sustained service interruptions to consumers. These events are of particular interest to us, because they help us understand the physical consequences caused by the integrity attacks.


\subsection{Simulation Methodology and Settings}

We use a distribution feeder specification \cite{gridlabd-spec},
which covers a moderately populated urban area and comprises 1405 houses, 2134 buses, 3314 triplex buses, 1944 transformers, 1543 overhead lines, 335 underground lines, and 1631 triplex lines. For this small-scale distribution feeder, locational prices are usually not applicable and hence all the houses are subject to the same price as discussed in Section~\ref{subsec:iso}. By leveraging the extensibility of GridLAB-D, we develop new modules that implement the CEO model for each single house, the price stabilization algorithm in Eq.~(\ref{eq:classical-control}), and the attack strategies. We measure the instantaneous power of the entire feeder at the root node. Its peak value over the previous pricing period is used as $d(\lambda_{k-1})$ in Eq.~(\ref{eq:classical-control}). As we focus on evaluating the physical consequences of attacks, we do not simulate the logistics of the attacks and assume that the adversary can gain access to the meters of his choosing. Specifically, if a house is not subject to attacks, it directly reads the real-time price from the {\em ISO module}; otherwise, it reads the price from an {\em adversary module} that modifies the price according to the attack models.
All the attacks are launched after the system has converged.

We adopt the CEO demand model for each single house,
where the parameters are drawn from normal distributions: $D_i \sim \mathcal{N}(7, 3.5^2)$ (unit: kW) and $\epsilon_i \sim \mathcal{N}(-0.8, 0.1^2)$. Under this setting, if the price is within $[10, 20]$, the per-house price-responsive demand is within $[0.65, 1.1]$ kW. To improve the realism of the simulations, we use the half-hourly total demand trace from March 1st to 22nd, 2013, of NSW, Australia, provided by AEMO \cite{AEMO} as the baseline load. The baseline load of a single house is set to be a scaled version of the real data trace. The resultant range of the per-house baseline load is $[0.276, 0.488]$ kW. Hence, when the price is within $[10, 20]$, the demand of a household is within $[0.9, 1.6]$ kW, which is consistent with the average demand of a household in reality \cite{household-use2}. In our simulations, the price is updated every half an hour, to be consistent with the setting of the demand data traces \cite{AEMO}. In each pricing period, the simulated demand remains constant. For the supply model, the settings obtained in Fig.~\ref{fig:supply-price} (i.e., $p=152$ and $q=4503$) are for the whole NSW region. They must be scaled down to fit the simulated feeder with 1405 houses. From the facts that there are 2.8 million households in NSW \cite{household-use2} and 57\% of AEMO's supply is for residential demand \cite[p.~15]{aemo-forecasting}, the two parameters are scaled as follows: $p=\frac{57\% \times 152}{2800000/1405}= 43.638 \times 10^{-3}$ and $q = \frac{57\% \times 4503}{2800000/1405} = 1.287$. Other default settings include: $T=0.5$, $\lambda_{\min} = 1$, $\lambda_{\max} = 200$, and $\eta=0.5$.

\subsection{Simulation Results}

\noindent {\bf Price stabilization:}
The first set of simulations evaluate the effectiveness of the direct feedback approach \cite{roozbehani2011volatility} and our control-theoretic price stabilization algorithm in Eq.~(\ref{eq:classical-control}). In the simulations, the direct feedback approach is unstable, where the price oscillates between $\lambda_{\min}$ and $\lambda_{\max}$. The total demand reaches 10 MW a few hours after the start of the simulation, and GridLAB-D reports that four distribution lines are overloaded.
Fig.~\ref{fig:sim-ours} plots the price and resultant demands under our price stabilization approach in Eq.~(\ref{eq:classical-control}). We can see that the price fluctuates slightly for a few hours after the start of the simulation, due to an inappropriate initial price. After the system converges, it can well adapt to the time-varying baseline load. The generation scheduling error is close to zero, which means that the clearing price is achieved. We also evaluate the impact of an inaccurate consumer model $w(\lambda) = \widetilde{D} \lambda^{\widetilde{\epsilon}}$ on the system performance. The evaluation results are consistent with our analysis in Section~\ref{subsec:inaccurate-w}. For instance, when $\eta=0.5$, the system keeps stable if i) $(D - \widetilde{D}) / D < 0.67$ and $\widetilde{\epsilon} = \epsilon$, or ii) $0.05 \epsilon < \widetilde{\epsilon} < 2.15 \epsilon$ and $\widetilde{D} = D$. Therefore, the price stabilization algorithm is resilient to inaccuracies of the estimated consumer model.

\begin{figure}
  \centering
    \includegraphics{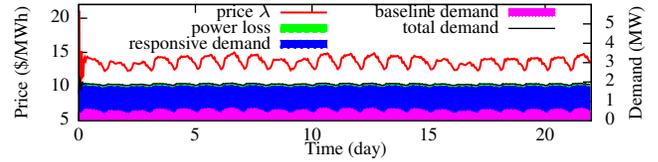}
    \vspace{-0.3in}
    \caption{Price stabilization in the absence of attack.}
  \label{fig:sim-ours}
\end{figure}

\begin{figure}
  \centering
  \includegraphics{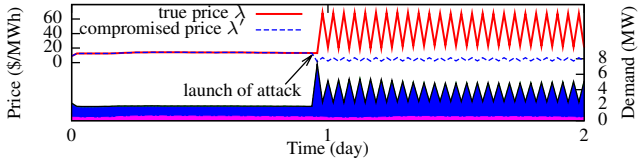}
    \vspace{-0.3in}
    \caption{Scaling attack ($\rho=65\%$, $\gamma=0.1$).}
    \label{fig:sim-scaling}
\end{figure}

\vspace{0.5em}
\noindent {\bf Scaling attack:}
Fig.~\ref{fig:sim-scaling} plots the true and compromised prices, as well as the breakdown of demand under the scaling attack. We can see that the price as well as the demand fluctuates severely. GridLAB-D reports excessive distribution line overload events after the launch of the attack. We also extensively evaluate the impact of the scaling attack with different settings of $\rho$ and $\gamma$. We use the standard deviation of the generation scheduling error after the launch of the attack, denoted by $\sigma(e)$, as the system volatility metric. A near-zero $\sigma(e)$ means convergence, while a considerably large $\sigma(e)$ means oscillation or divergence. Fig.~\ref{fig:scale_err_gamma} plots $\sigma(e)$ versus $\gamma$ under various settings of $\rho$. We can see that the system volatility increases with $\rho$ and decreases with $\gamma$.

\begin{figure}
  \centering
  \includegraphics{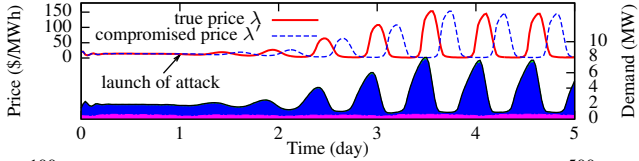}
  \includegraphics{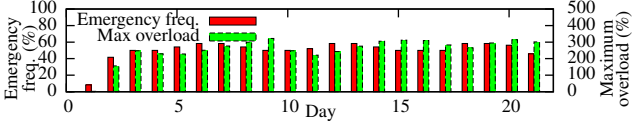}
  \vspace{-0.1in}
  \caption{Impact of delay attack ($\rho=100\%$, $\tau=9$). (a) Prices and demands; (b) emergencies.}
  \label{fig:sim-delay1}
\end{figure}

\vspace{.5em}
\noindent {\bf Delay attack:}
Fig.~\ref{fig:sim-delay1}(a) and Fig.~\ref{fig:sim-delay2}(a) show the evolution of price and the breakdown of demand under the delay attacks with different parameters. We also investigate the emergency events reported by GridLAB-D. The {\em overload} of a distribution line or a transformer is defined as the percentage of the exceeded current/power with respect to the rated value. Fig.~\ref{fig:sim-delay1}(b) and Fig.~\ref{fig:sim-delay2}(b) plot the {\em emergency frequency} and maximum overload in each day. The emergency frequency is defined as the ratio of the number of pricing periods with reported emergency events to the number of pricing periods per day (i.e., 48). In Fig.~\ref{fig:sim-delay1}, a small generation scheduling error caused by the time-varying baseline load will be amplified iteratively along the control loops, after the launch of the  attack. The overload can be up to 350\%. In practice, such a high overload will cause circuit breakers to open and hence regional blackouts. In Fig.~\ref{fig:sim-delay2}, the system appears to diverge and then converge again without causing any emergencies. However, it diverges again from the 12th day due to the changing baseline load, causing excessive emergencies. This illustrates the stealthiness of the delay attack that causes marginal system stability.
Lastly, we evaluate the impact of the delay attack with different settings of $\rho$ and $\tau$. The results are shown in Fig.~\ref{fig:delay_err_rho}. We can clearly see that when $\rho < 0.5$, the system remains stable, which is consistent with Proposition~\ref{PROP:HARD-DELAY-ATTACK}.


\begin{figure}
  \centering
  \includegraphics{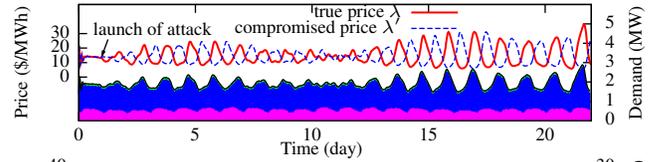}
  \includegraphics{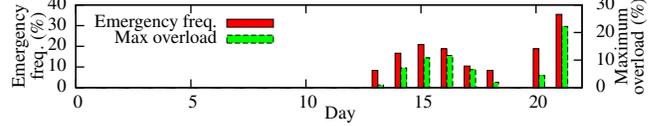}
  \vspace{-0.2in}
  \caption{Impact of delay attack ($\rho=65\%$, $\tau=24$). (a) Prices and demands; (b) emergencies.}
  \label{fig:sim-delay2}
  \vspace{-0.1in}
\end{figure}

\begin{figure}
  \centering
  \subfigure[Scaling attack]{
    \includegraphics{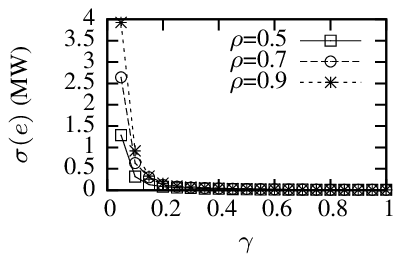}
    \label{fig:scale_err_gamma}
  }
  \subfigure[Delay attack]{
    \includegraphics{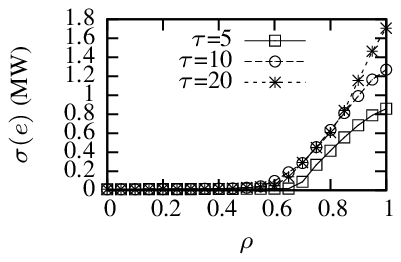}
    \label{fig:delay_err_rho}
  }
  \vspace{-0.1in}
  \caption{System volatility under attacks.}
  \vspace{-0.1in}
\end{figure}

\section{Conclusion}
\label{sec:conclude}

This paper systematically investigates the impact of scaling and delay attacks on the stability of RTP systems. We characterize the impact using a control-theoretic metric, {\em region of stability}.
We show that, to destabilize the RTP system, it is necessary for the adversary to reduce the price in a scaling attack or compromise more than half of the smart meters in a delay attack. We conduct trace-driven simulations to validate our analysis. The results of this paper improve our understanding of the security of RTP systems so that suitable defensive measures can be taken.

\subsection{Limitations and Future Work}


In this paper, we have made a few simplifying assumptions and have pointed them out. Care is thus required when applying our results to general real-world settings, e.g., large-scale smart grids, where locational prices are important. Particularly, Propositions~\ref{PROP:ALGORITHM} to \ref{PROP:HARD-DELAY-ATTACK} are based on the linearized abstract supply and demand models at a fixed operating point.\footnote{The numerical results and simulations use the algorithm in Eq.~(\ref{eq:classical-control}), which uses the current price as the operating point.} The security analysis in Propositions~\ref{PROP:SCALING-ROS} to \ref{PROP:HARD-DELAY-ATTACK} is based on the approximate aggregate demand model in Eq.~(\ref{eq:demand-under-attack}). These simplifications enable us to study the RTP susceptibility to malicious attacks under LTI settings. However, for systems that involve non-linear, time-variant, or stochastic components and complex market structures, the LTI assumptions may not hold and our analysis may lead to inaccuracies in characterizing such systems.

It is interesting for further research to address the following issues not considered in this paper. First, our analysis framework can be extended to address attacks that cannot be modeled as LTI systems. Second, for large-scale systems where locational prices are important, the system behavior may deviate from our analysis when the power network is congested due to price oscillations. Locational prices can be considered by integrating economic dispatch into the analysis. Third, this paper considers simplified system models while preserving the principles of RTP. Extensions are possible that consider various practical factors such as suppliers' ramp constraints, energy storage, load shifting, bidding markets, and ex-ante RTP with ex-post adjustments.

\section{Acknowledgements}

We thank our shepherd Dr.~Alvaro~A.~C\'{a}rdenas and the anonymous reviewers for providing valuable feedback on this work. This work was supported in part by Singapore's Agency for Science, Technology and Research (A$\star$STAR) under the Human Sixth Sense Programme, and in part by U.S. National Science Foundation under grant numbers CNS-0963715 and CNS-0964086.


\bibliographystyle{abbrv}
\bibliography{references}

\appendix

Denote $u_1\!=\! h \! + \! 1$, $u_2 \! = \! 2\eta \! + \! 2\eta(1 \! - \! \rho)h \! - \! h \! - \! 1$, $u_3 \! = \! 2\eta \rho h$, and $P(z) \! = \! u_1z^{\tau \! + \! 1} \! + \! u_2z^\tau \! + \! u_3$, where $u_1 > 0$ and $u_3 > 0$. Express any system pole (i.e., root of $P(z)=0$) in polar coordinates: $z = A(\cos \theta + i \sin \theta)$, where $A > 0$.
$P(z)=0$ can be rewritten as two equations: $A^\tau \left( u_1 A \cos(\tau+1)\theta + u_2 \cos\tau \theta \right) = -u_3$ and $A^\tau(u_1 A \sin(\tau+1)\theta + u_2 \sin\tau \theta) = 0$.
Adding the squares of the two equations yields $g(A) \! = \! 0$, where $g(A) \! = \! u_1^2 A^{2\tau\!+\!2} \! + \! 2u_1u_2\cos\theta A^{2\tau\!+\!1} \! + \! u_2^2 A^{2\tau} \! - \! u_3^2$. Thus, any pole satisfies $g(A)\!=\!0$.
We can verify that $g(1) > 0$ when $\rho \in (0,0.5]$. Moreover, $\dot{g}(A) = A^{2\tau} m(A)$, where $m(A) = u_1^2 (2\tau+2)A + 2u_1u_2(2\tau+1)\cos\theta + \frac{2\tau u_2^2}{A}$ is a convex function with its minimum at $A^* = \left| \frac{u_2}{u_1} \right| \sqrt{\frac{\tau}{\tau+1}}$. We can verify that $\left| \frac{u_2}{u_1} \right| < 1$ if $\rho \in (0, 0.5]$. Thus, $A^* < 1$ and the minimum of $m(A)$ for $A \ge 1$ is $m(1)$, which satisfies
$m(1) \! \ge \! u_1^2 ( 2\tau \! + \! 2) \! - \! 2u_1| u_2 |(2\tau \! + \! 1) \! + \! 2\tau \! u_2^2 \! = \! 2 ( u_1 \! - \! | u_2 |)(\tau(u_1 \! - \! | u_2 |) \! + \! u_1 )$.
If $u_2 < 0$, $u_1-|u_2|=2\eta + 2\eta(1-\rho)h > 0$; if $u_2 \ge 0$, $u_1 - |u_2| = 2(1-\eta) + 2h(1 - \eta(1-\rho)) > 0$. Hence, $m(A) > 0$ and $\dot{g}(A) > 0$ for $A \ge 1$. Recalling $g(1) > 0$, we have $g(A) > 0$ for $A \! \ge \! 1$. Hence, $A \! < \! 1$ for all poles and the system is stable.

\end{document}